%% file: main.tex
\documentclass[10pt]{article}

\usepackage{authblk}                                                          
\usepackage{geometry}                                                       
\usepackage[utf8]{inputenc}                                                
\usepackage{hyperref}                                                        
\usepackage{url}                                                                
\usepackage{booktabs}                                                       
\usepackage{amsfonts}                                                       
\usepackage{nicefrac}                                                         
\usepackage{microtype}                                                      
\usepackage{wrapfig}
\usepackage{amsmath,amsthm,amssymb,mathtools,graphicx,enumitem,hyphenat,float}
\usepackage{bbm}
\usepackage{subcaption}
\usepackage{accents, hyperref}
\usepackage{import}
\usepackage{algorithm}
\usepackage{algpseudocode}
\usepackage[sort]{natbib}
\setcitestyle{numbers,open={[},close={]}}
\title{Formatting Instructions For NeurIPS 2020}
\graphicspath{}
\setcitestyle{citesep={,}}
\usepackage{appendix}

\newcommand{\R}{\mathbb{R}}

\newcommand{\PCIS}{\mathcal{S}}   
\newcommand{\Xset}{\mathcal{X}}
\newcommand{\Uset}{\mathcal{U}}
\newcommand{\Vdec}{\lambda}       
\newcommand{\sig}{\sigma}
\newcommand{\mumu}{\mu}
\newcommand{\norm}[1]{\left\lVert #1 \right\rVert}

\newcommand{\printfnsymbol}[1]{%
  \textsuperscript{\@fnsymbol{#1}}%
}

\input{header.tex}

\begin{document}

\title{\bf Safe Exploration for Nonlinear Processes Using Online Gaussian Process Learning}
\author[1]{S.\ Tonini}
\author[1]{S.\ Rastegarpour}
\author[1]{H. R.\ Feyzmahdavian}
\author[2]{N.\ Bastianello}
\author[2]{K. H.\ Johansson}
\affil[1]{ABB Corporate Research, V\"aster\aa s, Sweden} 
\affil[2]{KTH Royal Institute of Technology, Stockholm, Sweden}
\maketitle

%
%
\begin{abstract}                
This paper proposes a safe data-driven control framework for nonlinear systems with partially known dynamics. The method ensures stability and constraint satisfaction during online learning, assuming only a stabilizable linear approximation of the process is available. Unmodeled nonlinear dynamics are captured by a Gaussian process residual learned in real time. Safety is enforced through a probabilistic control-invariant set derived from Lyapunov theory, guaranteeing high-probability stability. A convex quadratic program computes control inputs that maximize information gain while respecting probabilistic safety constraints. The framework provides finite-sample safety guarantees and allows adaptive expansion of the invariant set as uncertainty decreases. Numerical results validate the approach, demonstrating safe and informative exploration under model uncertainty.
\end{abstract}

%
%
\section{Introduction}
Data-driven and learning-based control have emerged as powerful paradigms for improving performance in complex dynamical systems where first-principles models are incomplete or uncertain. Using data collected from operation, these approaches can adapt to unknown dynamics and disturbances, offering improved accuracy and robustness. However, in safety-critical applications such as autonomous robotics, aerospace systems, and process industries, the benefits of data-driven learning must be reconciled with strict guarantees of stability and constraint satisfaction. The central challenge is therefore to design learning-based controllers that ensure safe exploration, that is, the ability to learn system behavior online without ever violating safety or stability requirements.

Conventional machine learning methods such as reinforcement learning and Bayesian optimization popularize the idea of learning by interaction, yet they typically neglect hard safety constraints during early exploration \citep{hewing2020learning}. This has motivated extensive research on safe learning for control, in which machine learning models are integrated within control-theoretic frameworks that enforce stability and constraint invariance. A common strategy is to exploit a priori system knowledge, such as approximate dynamics or stability certificates, to limit risk during learning \citep{Berkenkamp2016,Koller2018SafeMPC,wang2025robust, rastegarpour2024enhancing, rastegarpour2025adaptive, pawlak2025hybrid}. While these approaches achieve promising results, many rely on idealized assumptions: either the full nonlinear dynamics is known, or uncertainty is captured through conservative Lipschitz over-approximations. Such assumptions limit their practicality for real systems, where a local linearization is usually available and the true residual dynamics must be learned online.

The motivation for this work arises from \cite{prajapat2025safe}, which demonstrated guaranteed safe exploration for nonlinear systems with fully known dynamics. Their framework synthesizes exploratory inputs that maintain invariance of precomputed safety sets, thereby ensuring constraint satisfaction at all times. Although theoretically attractive, this approach requires exact knowledge of the process dynamics, an assumption rarely satisfied in practice. Most real-world systems are instead characterized by a nominal linear model, around which learning must occur safely to identify unmodeled nonlinear effects.

This paper addresses this limitation by proposing a safe data-driven control framework that guarantees stability and constraint satisfaction when only a stabilizable linear approximation of the nonlinear process is known. The approach augments the nominal linear model with a Gaussian process (GP) residual that is learned online, capturing unmodeled nonlinear dynamics. Safety is ensured through a probabilistic control-invariant set (PCIS) derived from Lyapunov theory, within which the closed-loop system remains stable with high probability. The resulting controller achieves safe and informative exploration by balancing excitation for learning with strict safety guaranties.

The main contributions of this paper are summarized as follows:
\begin{itemize}
\item A nominal-residual decomposition framework that combines a known linear model with an online GP residual to enable safe learning of unknown nonlinearities without requiring full model knowledge.
\item A Lyapunov-based PCIS construction that incorporates GP uncertainty to provide finite-sample probabilistic guarantees of stability and constraint satisfaction.
\item An optimization-based controller that solves a convex quadratic program (QP) at each step to maximize information gain while maintaining safety almost surely.
\end{itemize}

Compared with existing safe learning methods, the proposed framework achieves a favorable balance between tractability and conservatism. Unlike SafeMPC formulations \citep{Koller2018SafeMPC}, which propagate uncertainty over long horizons and may become overly conservative, the present approach enforces a one-step Lyapunov decrease with GP confidence bounds, yielding a convex QP without multi-step uncertainty blow-up. In contrast to guaranteed safe exploration approaches with known dynamics \citep{Prajapat2024}, the proposed method requires only a stabilizable linear model, while learning the nonlinear residual online with finite-sample confidence. Furthermore, compared to probabilistic invariance formulations \citep{Gao2021PCIS,Griffioen2023}, our method provides a practical control Lyapunov function-based that integrates directly into a real-time controller, with adaptive safety set expansion as the GP uncertainty diminishes.

In summary, this work develops a unified and implementable framework for safe exploration in partially known nonlinear systems, integrating GP-based uncertainty quantification with Lyapunov stability theory and convex optimization. The approach provides theoretical safety guarantees and is validated through simulation studies that demonstrate effective and stable exploration under model uncertainty.

The remainder of this paper is organized as follows. Section~\ref{sec:background} introduces the necessary preliminaries. Section~\ref{sec:problem} formulates the safe exploration problem. Section~\ref{sec:method} presents the proposed GP–PCI construction and controller synthesis. Section~\ref{sec:results} provides numerical results, and Section~\ref{sec:conclusion} concludes the paper.

\section{Background}\label{sec:background}
We review the three ingredients underlying our approach:
(i) Gaussian processes (GPs) to model unknown residual dynamics with calibrated uncertainty,
(ii) probabilistic control–invariant sets (PCIS) constructed via Lyapunov conditions, and
(iii) an optimisation–based controller that enforces safety while promoting informative excitation.

\subsection{Gaussian processes for system identification}
Gaussian processes provide a Bayesian, non-parametric prior over functions \citep{RasmussenWilliams2006GPML}.  Given input-output data $\mathcal D=\{(x_i,y_i)\}_{i=1}^N$, a GP posterior yields closed-form expressions for the predictive mean $\mu(\cdot)$ and standard deviation $\sigma(\cdot)$ at any query point.  Crucially, the uncertainty band $\mu\pm\beta\,\sigma$ contains the true function with high probability $1-\delta$ when the kernel is chosen appropriately and $\beta$ is tuned via information-theoretic regret bounds \citep{Srinivas2010GP-UCB}.  Unlike parametric regressors, GPs gracefully capture non-linearities with only a handful of data, an essential property for online learning in safety-critical systems where extensive exploration is infeasible.
A GP is a collection of random variables such that any finite subset is jointly Gaussian. With mean \(m(x)\) and kernel \(k(x,x')\),
\[
f(x)\sim \mathcal{GP}\!\bigl(m(x),\,k(x,x')\bigr).
\]
Given training data \(\mathcal D=\{(x_i,y_i)\}_{i=1}^n\) and noisy observations \(y=f(x)+\varepsilon\), \(\varepsilon\sim\mathcal N(0,\sigma_n^2)\), the joint prior over \(y\) and the test value \(f_\ast=f(x_\ast)\) are
\[
\begin{bmatrix} y \\ f_\ast \end{bmatrix}
\sim\mathcal N\!\left(
\begin{bmatrix} m(X) \\ m(x_\ast) \end{bmatrix},
\begin{bmatrix}
K(X,X)+\sigma_n^2 I & K(X,x_\ast)\\
K(x_\ast,X) & K(x_\ast,x_\ast)
\end{bmatrix}
\right),
\]
where \(K(\cdot,\cdot)\) is the kernel Gram matrix. Conditioning yields the standard predictive equations
\begin{align}
\mu(x_\ast) &= m(x_\ast) + K(x_\ast,X)\bigl[K(X,X)+\sigma_n^2 I\bigr]^{-1}\!\bigl(y-m(X)\bigr),\label{eq:gp_mean}\\
\sigma^2(x_\ast) &= K(x_\ast,x_\ast) - K(x_\ast,X)\bigl[K(X,X)+\sigma_n^2 I\bigr]^{-1}\!K(X,X_\ast).\label{eq:gp_var}
\end{align}

The pair \((\mu,\sigma)\) provides both a point estimate and calibrated uncertainty. In our setting, the residual dynamics \(g(\cdot)\) are modeled by a GP and we enforce safety using a high-probability envelope \(|g(x)-\mu(x)|\le\beta_t\,\sigma(x)\) inside the Lyapunov-based predicate (Sec.~\ref{sec:method}).

The main limitations of GP is that naive training scales as \(\mathcal O(n^3)\) and prediction as \(\mathcal O(n^2)\) due to Cholesky solves on \(K(X,X)\), which motivates sparse / inducing point approximations for online operation (see, e.g., \citep{RasmussenWilliams2006GPML}).

\subsection{Control-invariant sets and Lyapunov functions}
A control-invariant set $\mathcal C$ is a subset of the state space such that for every $x\in\mathcal C$ there exists an admissible control $u$ keeping the closed-loop trajectory inside $\mathcal C$ for all future time.  When a continuously differentiable Lyapunov function $V(x)$ satisfies $\dot V\le0$ in $\mathcal C$, the sublevel set $\{x\mid V(x)\le\alpha\}$ is control-invariant.  Classical results assume perfectly known dynamics and, therefore, deterministic invariance conditions.  In the presence of modeling errors, deterministic constraints become overly conservative.  Recent research blends Lyapunov theory with stochastic certificates, yielding probabilistic control-invariant sets (PCI) that tolerate bounded uncertainty while still guaranteeing safety with high probability \cite{wang2018safe,Gao2021PCIS}.

Integrating GP uncertainty bounds into Lyapunov conditions enables data-driven refinement of the safe set: as the model improves, the GP variance $\sigma$ shrinks, and the PCI expands, unlocking progressively richer exploration without sacrificing guarantees.

\subsection{Confidence schedules and calibration}\label{subsec:beta}
We use the finite-sample GP-UCB schedule (e.g., \cite{Srinivas2010GP-UCB})
\[
\beta_t \;=\; \sigma_n \sqrt{2\big(\gamma_{t-1} + 1 + \ln(1/\delta_t)\big)}\ +\ B,
\]
where $\gamma_{t-1}$ is the maximum information gain up to $t-1$, $B$ bounds the RKHS norm of $g$, and $\sigma_n^2$ is the observation noise variance. In practice we upper bound $\gamma_{t-1}$ conservatively for the chosen kernel.

We estimate a scalar $\gamma^\star\ge 1$ on a validation split such that the empirical coverage of $\mumu\pm 1.96\sqrt{\gamma^\star}\sig$ is at least $95\%$. We then use $\tilde\sig=\sqrt{\gamma^\star}\sig$ inside the optimization problem \eqref{eq:pcis_constraint}. This preserves the guarantee provided calibration uses held-out data and is recomputed only between runs.

We keep the QP convex by using a linearized variance-seeking term:
\[
\min_{u,s}\ (u-u_{lin})^\top R_s (u-u_{lin}) + \rho s - \alpha\, {w}^\top \sig_t(x),
\]
with $\alpha\ge 0$ and ${w}\in\R^n_{\ge 0}$ selecting directions of interest (e.g., basis aligned with largest posterior variance). Since $\sig_t$ is state dependent, we use the current $\sig_t(x_k)$ as a constant weight at step $k$, keeping the program quadratic. 

\section{Problem Framework}\label{sec:problem}

We consider a control-affine nonlinear system
\begin{equation}\label{eq:plant}
\dot x = f(x,u) = Ax + Bu + g(x),
\end{equation}
where \((A,B)\) is a stabilisable known linearisation around an operating point and \(g:\mathcal X\to\mathbb R^{n}\) is an unknown smooth residual capturing unmodelled dynamics. A stabilising linear feedback \(u_{\mathrm{lin}}(x)=-Kx\) with Lyapunov function \(V(x)=x^\top P x\) is available for the nominal pair \((A,B)\).
The controller runs with ZOH at period \(\Delta t>0\). Let \(x_k:=x(k\Delta t)\), \(u_k:=u(k\Delta t)\). Noisy measurements are
\begin{equation}
y_k = Cx_k + \nu_k,\qquad \nu_k\sim\mathcal N(0,R).
\end{equation}
When only \((y_k)\) are measured, a state estimate \(\hat x_k\) (observer or filtered GP prediction) is used instead of \(x_k\) in the safety predicate.
State and input constraints are known and encode safety and actuation limits:
\begin{align}\label{eq:box-constraints}
\mathcal X &= \{\,x \in \mathbb{R}^n \;:\; x_{\min}\le x \le x_{\max}\,\},\\[-2pt]
\mathcal U &= \{\,u \in \mathbb{R}^m \;:\; u_{\min}\le u \le u_{\max}\,\}. 
\end{align}

\subsection{Residual model and prior information}
The unknown residual \(g(\cdot)\) is modeled by a Gaussian Process with kernel \(k\). At time \(t\) the GP posterior provides mean \(\mu_t(\cdot)\) and std.\ \(\sigma_t(\cdot)\). We use a high-confidence envelope with scale \(\beta_t>0\):
\begin{equation}\label{eq:hp-envelope}
\Pr\big(\|g(x)-\mu_t(x)\|_2 \le \beta_t\,\sigma_t(x)\big)\ \ge\ 1-\delta,
\qquad \forall x\in\mathcal X,
\end{equation}
where \(\delta\in(0,1)\) is a user-chosen risk level (Section~\ref{subsec:risk}). Standard choices for \(\beta_t\) follow GP-UCB schedules based on information gain.

\subsection{Safety notion (probabilistic invariance)}
Safety is defined as recursive feasibility of state-input constraints with high probability:
\begin{equation}\label{eq:safety}
\Pr\!\Big(x_k\in\mathcal X,\ u_k\in\mathcal U,\ \forall k\ge 0\ \Big)\ \ge\ 1-\delta.
\end{equation}
We certify safety by constructing a probabilistic control-invariant set (PCIS) \(\mathcal S\subseteq\mathcal X\) such that \(x_0\in\mathcal S\) implies \(x_k\in\mathcal S\) for all \(k\) with probability \(1-\delta\).

\subsection{Exploration objective}
Let \(\mathcal D_t=\{(x_i,u_i,y_i)\}_{i=1}^{N_t}\) be the data collected up to time \(t\). The goal is to design a feedback policy \(\pi\) that
\begin{equation}\label{eq:objective}
\max_{\pi}\ \ \mathcal I(\mathcal D_T)\quad
\text{s.t.}\ \ \eqref{eq:plant},\ \eqref{eq:box-constraints},\ \eqref{eq:safety},
\end{equation}
where \(\mathcal I(\cdot)\) measures information (e.g., trace reduction of the GP posterior covariance, mutual information, or prediction RMSE on a validation grid). In practice we use a variance-seeking objective filtered by a safety layer (Section~\ref{sec:method}).

\subsection{Risk allocation and notation}\label{subsec:risk}
We adopt a single risk budget \(\delta\) for invariance; equivalently, one can distribute per-step or per-constraint budgets \(\{\delta_k\}\) s.t.\ \(\sum_k \delta_k\le \delta\) (union bound). Bold symbols denote vectors; \(\|\cdot\|\) is the Euclidean norm; \(\lambda>0\) denotes the desired Lyapunov decay rate used in the PCIS condition.

\section{Methodology}\label{sec:method}

Consider the control‑affine continuous‑time system
\begin{equation}\label{eq:nonlin}
  \dot x = f(x,u) = Ax + Bu + g(x), \quad x\in\mathcal X,\; u\in\mathcal U,
\end{equation}
where $(A,B)$ is a stabilisable linearisation, and $g:\mathcal X\!\to\R^n$ is unknown but \emph{smooth}.  Hard constraints $x\in\mathcal X$ and $u\in\mathcal U$ encode actuator limits and safety envelopes.  We assume a stabilising linear controller $u_{\mathrm{lin}}(x)$ keeps the origin exponentially stable for the linear model.

\textbf{Objective.} Design an exploration policy that (i) collects informative data to learn $g$ and (ii) never violates $\mathcal X$ nor $\mathcal U$ with probability at least $1-\delta$.

\subsection{GP residual model}
We model the residual $g$ as a zero‑mean GP,
\begin{equation}
  g \sim \mathcal{GP}\bigl(0,\,k(\cdot,\cdot)\bigr),
\end{equation}
updated online with data $\mathcal D_t=\{(x_i,g_i)\}_{i=1}^{N_t}$.  Standard GP inference yields a posterior mean $\mu_t(x)$ and standard deviation $\sigma_t(x)$. With high‑probability confidence parameter $\beta_t$ we have
\begin{equation}\label{eq:gp_conf}
  \Pr\bigl(\|g(x)-\mu_t(x)\|_2 \le \beta_t\sigma_t(x)\bigr) \ge 1-\delta.\!
\end{equation}

\subsection{Probabilistic control‑invariant set}
Let $V(x)=x^\top P x$ be the Lyapunov function of a stabilising LQR for $(A,B)$.  We define the \emph{GP‑aware} PCI
\begin{align}\label{eq:PCI}
  \ PCI &:= \Bigl\{x\in\mathcal X\,\big|\, \exists u\in\mathcal U:\ 
    \nabla V\!\cdot\!\bigl(Ax+Bu+\mu_t(x)\bigr) \\
    &\hphantom{:= \Bigl\{x\in\mathcal X\,}\quad+\beta_t\sigma_t(x)\|\nabla V\| + \lambda V(x) \le 0 \Bigr\},
\end{align}

which contracts under the worst‑case GP error.  By construction, trajectories starting in $PCI$ remain there with probability $1-\delta$.

\subsection{Safe excitation controller}
At each time step, we solve the following QP problem
\begin{equation}\label{eq:QP}
\begin{aligned}
\min_{u,\,s\ge 0}\;\; & \|u-u_{\mathrm{lin}}(x)\|_2^2 + \rho\, s \\
\text{s.t.}\;\; 
& \nabla V(x)^\top\!\big(Ax + Bu + \mu_t(x)\big) \\
&\quad +\, \beta_t\,\sigma_t(x)\,\|\nabla V(x)\|_2 \;+\; \lambda\,V(x) \;\le\; s,\\
& u_{\min} \le u \le u_{\max}.
\end{aligned}
\end{equation}

Slack~$s$ preserves feasibility, while the cost keeps $u$ close to $u_{\mathrm{lin}}$ whenever safety allows.
The scalar $\rho>0$ penalizes the constraint slack $s$ so that $s=0$ is preferred whenever the CLF inequality is satisfiable without relaxation. Intuitively, large $\rho$ makes any violation expensive and drives $s\!\to\!0$ except when strictly necessary (e.g., under large model uncertainty or tight bounds). The solution ensures $x_{t+1}\in PCI$ almost surely.

\subsection{Assumptions and guarantee}\label{subsec:guarantee}
We design and analyze in discrete time with ZOH.

\textbf{A1 (Nominal CLF).} There exist $P\!=\!P^\top\!\succ 0$, $\Vdec\!>\!0$, and a stabilizing nominal feedback $u_{lin}(x)\!=\!-Kx$ for $(A,B)$ such that 
\[
V(x)=x^\top P x,\quad V(x^+_{\text{nom}})-V(x)\le -\Vdec V(x)
\]
for the nominal map $x^+_{\text{nom}}=A_dx+B_du_{lin}(x)$ in a neighborhood of the origin.

\textbf{A2 (Residual model).} The true residual $g:\Xset\to\R^n$ satisfies the GP-UCB envelope componentwise on $\Xset$,
\[
|g_i(x)-\mumu_{t,i}(x)| \le \beta_t\, \sig_{t,i}(x),\quad i=1,\dots,n,
\]
with probability at least $1-\delta_t$ (Sec.~\ref{subsec:beta}). The posterior is updated online with data $(x_k,u_k)$.

\textbf{A3 (Discretization bound).} The continuous-time closed-loop with ZOH is discretized with step $\Delta t$; the local truncation error on $V$ is bounded as
\(
|V(x_{k+1})-V(x_k)-\Delta V_{\text{Euler}}(x_k,u_k)| \le \eta\,\Delta t^2
\)
for $x$ in the region of interest (standard for RK/linearization).

\textbf{A4 (State estimate).} If $x$ is not measured, $\hat x$ satisfies $\norm{x-\hat x}\le \epsilon_x$ with known bound (absorbed into margins).

\medskip

\textbf{QP safety constraint.} At time $k$, the controller solves
\begin{subequations}\label{eq:pcis_constraint}
\begin{align}
\min_{u,\,s\ge0}\;&(u-u_{\text{lin}}(x))^\top R_s (u-u_{\text{lin}}(x))+\rho s,\\
\text{s.t.}\;&V(Ax{+}Bu{+}\mu_t(x)){-}V(x)\notag\\
&\ +\sum_{i=1}^n \beta_t \sigma_{t,i}(x)\,\bigl|[\nabla V(x)]_i\bigr|
+\eta \Delta t^{2}\ \le\ -\Vdec V(x)+s,\\
&u\in\mathcal U.
\end{align}
\end{subequations}

The sum uses a componentwise envelope; equivalently use an induced-norm bound
$\beta_t \norm{\sig_t(x)}_\infty \norm{\nabla V(x)}_1$.

\begin{theorem}[Probabilistic invariance]\label{thm:pcis}
Let $\PCIS_t := \{x\in\Xset \mid \exists u\in\Uset \text{ satisfying } \eqref{eq:pcis_constraint} \text{ with } s=0 \}$. 
Under A1–A4, if $x_k\in\PCIS_t$ and the QP \eqref{eq:pcis_constraint} is feasible at $k$, then with probability at least $1-\delta_t$,
\[
x_{k+1}\in\PCIS_t\quad \text{and}\quad x_j\in\Xset,\ u_j\in\Uset\ \ \forall j\ge k.
\]
If a per-step risk budget $\delta_k$ is used with $\sum_{k=0}^{T-1}\delta_k\le \delta$, then
$\Pr(x_j\in \Xset, u_j\in\Uset,\ \forall j\le T) \ge 1-\delta$.
\end{theorem}

\begin{proof}[Sketch of proof]
The constraint enforces
\begin{equation}
V(x_{k+1})-V(x_k)\le -\Vdec V(x_k) \end{equation} for all residuals consistent with the GP-UCB envelope and discretization error, hence $V$ is nonincreasing within $\PCIS_t$ and $x_{k+1}$ stays in the sublevel set (invariance). By Assumption~A2 the envelope holds with probability at least \(1-\delta_t\); a union bound across time
yields the horizon guarantee. \qedhere
\end{proof}

\begin{algorithm}[t]
\small
\caption{Online Safe Exploration with GP-PCIS}
\label{alg:gppcis}
\begin{algorithmic}[1]
\Require Nominal $(A,B)$, $u_{lin}(x)=-Kx$, $V(x)=x^\top P x$, constraints $\mathcal X,\mathcal U$
\Require GP prior $(m,k)$, risk schedule $\{\beta_t,\delta_t\}$, (optional) calibration $\gamma^\star\!\ge\!1$
\State $\mathcal D \gets \emptyset,\ t\gets0$
\While{experiment running}
  \State Measure/estimate state $x_t$
  \State $(\mu_t,\sigma_t)\gets \textsc{GP\_Posterior}(\mathcal D)$
  \State $\tilde\sigma_t(x)\gets \sqrt{\gamma^\star}\,\sigma_t(x)$ \Comment{if calibrated}
  \State \textbf{Safety margin at $x_t$:}
  \State $g_{\text{safety}}(u)\gets \nabla V(x_t)^\top\!\big(Ax_t+Bu+\mu_t(x_t)\big)$
  \State \hspace{10mm}$+\,\beta_t\,\|\nabla V(x_t)\|_2\,\tilde\sigma_t(x_t)+\lambda V(x_t)$
  \State \textbf{Safe excitation (CLF--QP):}
  \State $(u_t^\star,s_t^\star)\gets\arg\min_{u,\,s\ge0}\ \|u-u_{lin}(x_t)\|_2^2+\rho s$
  \State \hspace{18mm}\text{s.t. } $g_{\text{safety}}(u)\le s,\quad u\in\mathcal U$
  \State Apply $u_t^\star$; obtain $x_{t+1}$; \ \ $\mathcal D\gets \mathcal D\cup\{(x_t,u_t^\star,\text{obs})\}$
  \State (Optional) Update GP hyperparameters / $\gamma^\star$ every $K$ steps
  \State $t\gets t+1$
\EndWhile
\end{algorithmic}
\end{algorithm}

\section{Benchmarks and Experimental Setup}\label{sec:setup}

We tested the proposed GP–PCIS framework on two
benchmarks of increasing complexity. For each case we
state: (i) dynamics and operating point; (ii) constraints and signals;
(iii) initialization; and (iv) what is evaluated.

\subsection{Polynomial system (2D)}\label{subsec:poly}
Dynamics:
$\dot x_1 = x_2,\quad \dot x_2 = -2x_2 + x_2^2 + u$. \\
Linearisation at the origin:
$A=\begin{bmatrix}0&1\\0&-2\end{bmatrix}$,
$B=\begin{bmatrix}0\\1\end{bmatrix}$. \\
Nominal LQR: $Q=0.1I$, $R=0.1$ (gain $K$ as in Sec.~\ref{sec:method});
sampling $\Delta t=0.01\,\mathrm{s}$.
Constraints: state band on $x_2$ (velocity) and input bounds $u\in[u_{\min},u_{\max}]$.

\subsection{Three-tank water-level process (MIMO)}\label{subsec:tank}
States/inputs: $x=[h_1,h_2,h_3]^\top$, $u=[v_1,v_2,v_3]^\top$.
Truth model: Torricelli flows with inter-tank orifices; see (\ref{eq:h1})–(\ref{eq:h3}).
Nominal model: stabilizable linearisation $(A,B)$ about $(h^\star,v^\star)$; residual $g(\cdot)$ learned by a GP.
Sampling: $\Delta t=1\,\mathrm{s}$ (ZOH); discretisation via RK4 for simulation.
Constraints: level bands $h_i\in[h_i^{\min},h_i^{\max}]$, valve saturation.
Parameters: see Table~\ref{tab:plant_params}; full list in the appendix if needed.
This MIMO, coupled, constrained plant aims to test certification tractability and real-time control.

%
%
\section{Simulation Results}
\label{sec:results}

The effectiveness of combining GP regression with PCIS is evaluated in terms of its ability to enable informative data collection under safety constraints. For each benchmark, the following aspects are assessed:  
(i) GP model accuracy;  
(ii) construction and evolution of the PCIS;  
(iii) performance of an unsafe baseline; and  
(iv) the proposed GP--PCIS safe exploration scheme.

The evaluation considers four categories of metrics:  
\begin{enumerate}
    \item Learning performance: root-mean-square error (RMSE), coefficient of determination ($R^2$), and log-marginal likelihood (LL);
    \item Uncertainty quantification: empirical coverage of nominal $95\%$ confidence intervals and mean predictive standard deviation $\bar{\sigma}$;
    \item Safety: number of state-constraint violations, minimum distance to constraint boundaries, and certified PCIS volume $\mathrm{vol}(\mathcal{S})$;
    \item Computational efficiency: training and prediction times, and peak memory usage.
\end{enumerate}
Exploration behavior is further characterized by the posterior variance and the growth of the certified safe set.

\subsection{Polynomial Benchmark}
\label{sec:results:poly}
A low-dimensional polynomial system is employed to examine three fundamental aspects:  
(i) the impact of kernel misspecification;  
(ii) the calibration of predictive uncertainty; and  
(iii) the conservatism of the PCIS certificate under equal risk.  
This benchmark isolates these core behaviors to provide a controlled assessment of model misspecification, uncertainty quantification, and safety certification conservatism.

\subsection{GP Kernel Comparison}
\label{sec:kernel-comparison}
Identical training and test datasets are used for all kernel configurations. Visualizations are generated on the slice $x_1{=}0$, $x_2 \in [-5,\,5]$. Training and evaluation on the same dataset across kernels reveal consistent trends, as summarized in Table~\ref{tab:kernel-metrics}. The polynomial kernel of degree~2 closely matches the quadratic residual (oracle reference).  
Models from the RBF family exhibit comparable mean accuracy ($\mathrm{RMSE} \approx 0.227$, $R^2 \approx 0.912$) but display severe under-coverage ($\sim 34\%$) with small mean predictive standard deviation $\bar{\sigma}$. In contrast, the Matérn-$5/2$ kernel yields higher prediction error but improved coverage ($74.5\%$). Computation time and memory usage differ marginally across kernels (Table~\ref{tab:kernel-cost}). In practice, the RBF-ARD kernel provides the best overall accuracy among the tested models.

\begin{table}[t]
  \centering
  \caption{Performance of Gaussian Process kernels for residual modeling of the two-dimensional polynomial system. 
  ``Coverage'' denotes the empirical fraction of test points within the nominal $95\%$ confidence intervals, and $\bar{\sigma}$ denotes the mean predictive standard deviation. 
  Best results are shown in \textbf{bold}.}
  \label{tab:kernel-metrics}
  \setlength{\tabcolsep}{3pt}
  \begin{tabular}{lccccc}
    \toprule
    \textbf{Kernel} & \textbf{RMSE} & $\boldsymbol{R^2}$ & \textbf{Log-lik.} & \textbf{Coverage} & $\boldsymbol{\bar{\sigma}}$ \\
    \midrule
    RBF                 & 0.2275 & 0.9124 & 571.36 & 34.0\% & 0.0428 \\
    Matérn\,3/2         & 0.8433 & 0.5541 & 550.58 & 27.0\% & 0.0869 \\
    Matérn\,5/2         & 0.3754 & 0.7615 & 566.04 & 74.5\% & 0.1165 \\
    RBF+Matérn\,3/2     & 0.2275 & 0.9123 & 571.36 & 34.0\% & 0.0428 \\
    Periodic+RBF        & 0.2277 & 0.9122 & 571.36 & 34.0\% & 0.0428 \\
    Linear+RBF          & 0.2275 & 0.9124 & 571.36 & 34.0\% & 0.0428 \\
    RBF (ARD)           & 0.2274 & 0.9125 & 571.36 & 34.0\% & 0.0428 \\
    Polynomial (deg.\,2)& \textbf{0.0000} & \textbf{1.0000} & \textbf{580.76} & \textbf{100.0\%} & 0.0038 \\
    \bottomrule
  \end{tabular}
\end{table}

\begin{figure}[t]
  \centering
  \includegraphics[width=\linewidth]{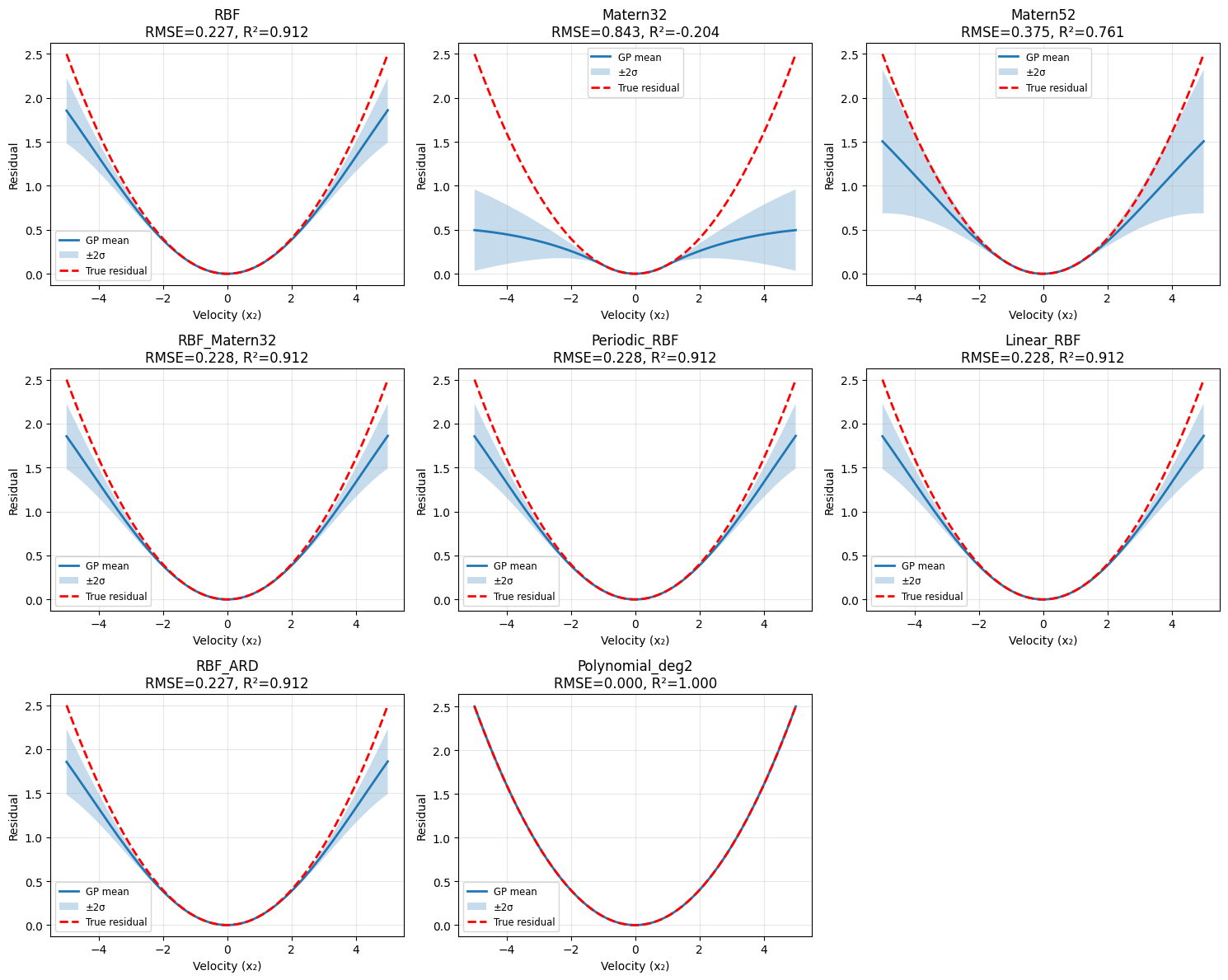}
  \caption{Kernel comparison of 2D Polynomial system along $x_1{=}0$. GP mean (blue) and $\pm 2\sigma$ band vs.\ true residual (red).}
  \label{fig:kernel-comparison}
\end{figure}

\begin{table}[t]
  \centering
  \caption{Polynomial benchmark computational cost and model complexity. Times is reported in seconds, and memory in MB.}
  \label{tab:kernel-cost}
  \begin{tabular}{lcccc}
    \toprule
    {Kernel} & Train & Predict & Total & Peak Mem. \\
    \midrule
    RBF                & 0.825 & 0.072 & 0.896 & 652 \\
    Matérn~3/2         & 0.877 & 0.057 & 0.934 & 657 \\
    Matérn~5/2         & 0.799 & 0.060 & 0.860 & 664 \\
    RBF+Matérn~3/2     & 1.002 & 0.087 & 1.090 & 675 \\
    Periodic+RBF       & 1.142 & 0.097 & 1.240 & 688 \\
    Linear+RBF         & 0.864 & 0.072 & 0.936 & 695 \\
    RBF–ARD            & 0.690 & 0.058 & 0.748 & 703 \\
    Polynomial (deg.\,2) & \textbf{0.634} & \textbf{0.054} & \textbf{0.688} & 705 \\
    \bottomrule
  \end{tabular}
\end{table}

\subsection{Full GP versus Sparse GP}
\label{subsec:poly_full_vs_sparse}

The polynomial benchmark is used to compare the exact GP model with the sparse FITC--SGPR approximation, using $M \in \{10, 20, 30, 40\}$ inducing points.

\begin{table}[t]
  \centering
  \caption{Comparison between exact GP and sparse FITC--SGPR models on the polynomial system. 
  ``Coverage'' denotes the empirical fraction of test points within the nominal $95\%$ confidence intervals.}
  \label{tab:gp-car}
  \setlength{\tabcolsep}{3pt}
  \begin{tabular}{lccccc}
    \toprule
    \textbf{Model} & \textbf{RMSE} & $\boldsymbol{R^2}$ & \textbf{Coverage (\%)} & \textbf{Time (s)} & $\boldsymbol{M}$ \\
    \midrule
    Full GP             & 0.4128 &  0.6719  & 90.0 & 0.27 & 80 \\
    Sparse GP ($M{=}10$) & 0.6671 &  0.1428  & 90.0 & 0.01 & 10 \\
    Sparse GP ($M{=}20$) & 3.2307 & -19.1014 &  5.0 & 0.01 & 20 \\
    Sparse GP ($M{=}30$) & 2.3529 &  -9.6625 & 15.0 & 0.01 & 30 \\
    Sparse GP ($M{=}40$) & 1.8136 &  -5.3349 & 10.0 & 0.01 & 40 \\
    \bottomrule
  \end{tabular}
\end{table}

The sparse FITC--SGPR models achieve approximately $30$--$40\times$ faster training times compared with the full GP. The configuration with $M{=}10$ offers the best trade-off between computational efficiency and predictive accuracy, although its performance remains below that of the full GP. In terms of safety, empirical coverage below $95\%$ indicates the need for variance calibration and/or an increased confidence parameter~$\beta$ prior to employing $\sigma$ within the PCIS framework; otherwise, the resulting certificates may become unreliable.

\subsection{Unsafe vs. safe exploration}
\label{subsec:poly_safe_vs_unsafe}
Without the safety layer, even stabilising LQR leaves $S(\alpha)$ (Fig.~\ref{fig:lqr_unsafe}). 
\begin{figure}[ht]
  \centering
  \includegraphics[width=0.8\linewidth]{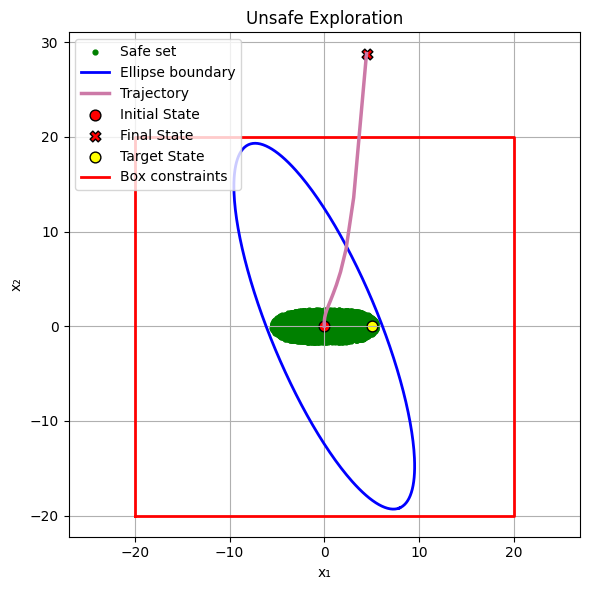}
  \caption{2D toy-\emph{unsafe} LQR tracking without safety layer: trajectory (red) exits the certified ellipsoid $S(\alpha)$ and box constraints reaching the target (yellow).}
  \label{fig:lqr_unsafe}
\end{figure}

\begin{figure*}[ht]
  \centering
  \captionsetup[subfigure]{justification=centering}
  \begin{subfigure}[t]{.24\textwidth}\centering
    \includegraphics[width=\linewidth]{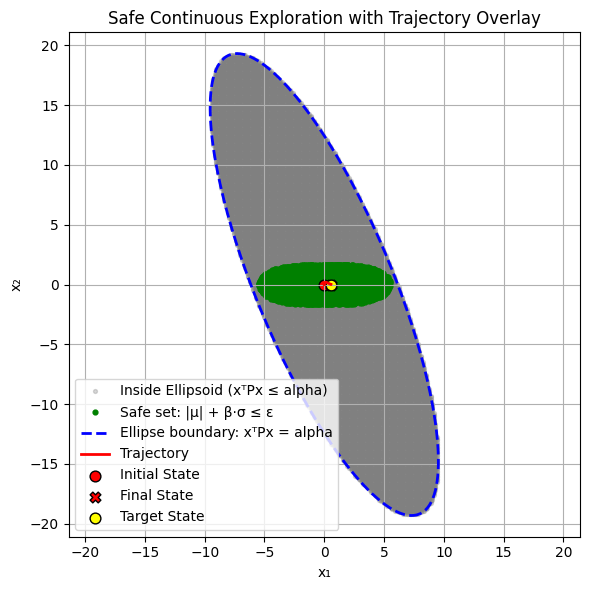}\subcaption{$t{=}0$}
  \end{subfigure}
  \begin{subfigure}[t]{.24\textwidth}\centering
     \includegraphics[width=\linewidth]{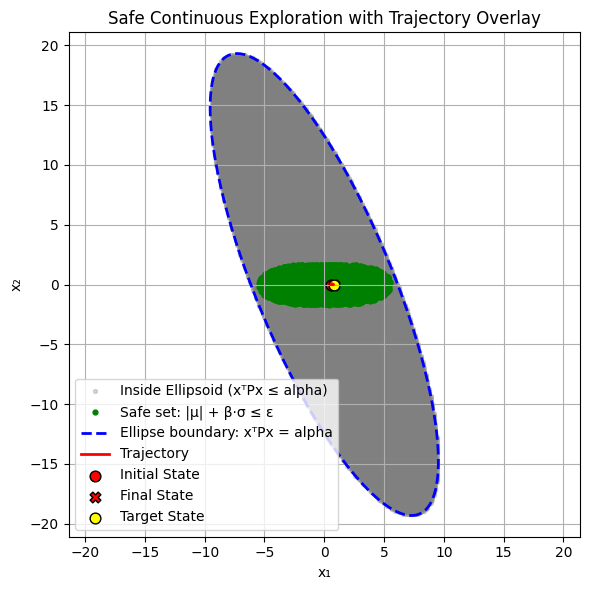}\subcaption{$t{=}1$}
   \end{subfigure}
  \begin{subfigure}[t]{.24\textwidth}\centering
    \includegraphics[width=\linewidth]{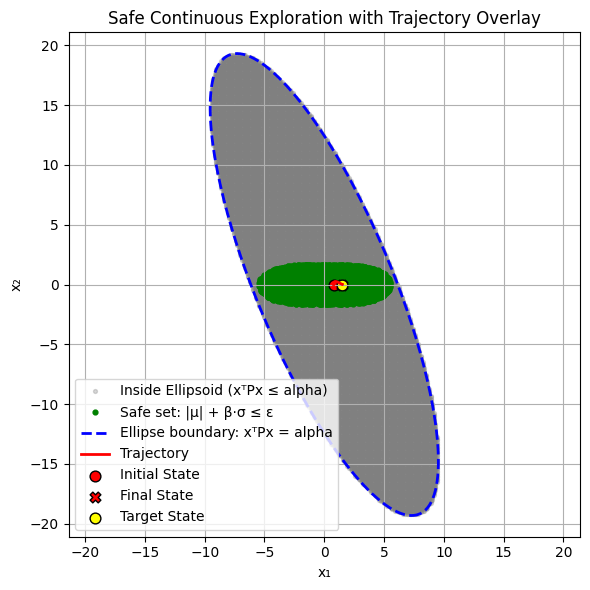}\subcaption{$t{=}2$}
  \end{subfigure}
  \begin{subfigure}[t]{.24\textwidth}\centering
    \includegraphics[width=\linewidth]{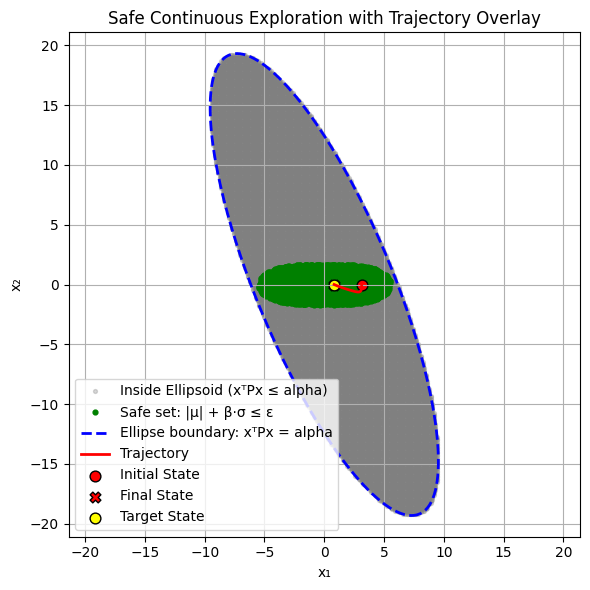}\subcaption{$t{=}12$}
  \end{subfigure}
  \caption{Safe GP-PCIS exploration on the 2D system: trajectory (red) remains inside PCIS (green ellipse).}
  \label{fig:safe-2d-sequence}
\end{figure*}

With our framework GP-PCIS, trajectories remain inside PCIS and the GP-safe region; as $\sigma(\cdot)$ shrinks, the certified admissible region expands (Tab.~\ref{tab:iter-metrics-2d}). 

\begin{table}[H]
  \centering
  \caption{Safe set size and GP performance per iteration (2D case).}
  \label{tab:iter-metrics-2d}
  \setlength{\tabcolsep}{2pt}
  \small
  \begin{tabular}{r r r r r r r r}
    \toprule
    \textbf{Iter} & \(|\mathcal{S}|\) & \textbf{RMSE} & \(\mathbf{MAE}\) & \(\mathbf{R^2}\) & \textbf{Coverage} & \(\bar{\sigma}\) & \textbf{Train pts} \\
    \midrule
     1 & 3993 & 1.1131 & 0.6782 & 0.8624 & 37.0\%  & 0.1772 &  100 \\
     2 & 3867 & 1.1131 & 0.6782 & 0.8624 & 37.0\%  & 0.1772 &  400 \\
     3 & 3900 & 1.1131 & 0.6782 & 0.8624 & 37.0\%  & 0.1772 &  700 \\
     4 & 4958 & 0.2325 & 0.1395 & 0.9940 & 100.0\% & 0.1172 & 1000 \\
     5 & 5092 & 0.2325 & 0.1395 & 0.9940 & 100.0\% & 0.1172 & 1300 \\
     6 & 5042 & 0.2325 & 0.1395 & 0.9940 & 100.0\% & 0.1172 & 1600 \\
     7 & 5095 & 0.0967 & 0.0578 & 0.9990 & 100.0\% & 0.0838 & 1900 \\
     8 & 5132 & 0.0967 & 0.0578 & 0.9990 & 100.0\% & 0.0838 & 2200 \\
     9 & 5055 & 0.0967 & 0.0578 & 0.9990 & 100.0\% & 0.0838 & 2500 \\
    10 & 5076 & 0.0293 & 0.0174 & 0.9999 & 100.0\% & 0.0520 & 2800 \\
    11 & 5217 & 0.0293 & 0.0174 & 0.9999 & 100.0\% & 0.0520 & 3100 \\
    12 & 5098 & 0.0293 & 0.0174 & 0.9999 & 100.0\% & 0.0520 & 3400 \\
    \bottomrule
  \end{tabular}
\end{table}

Table~\ref{tab:iter-metrics-2d} reports quantitative evolution per-iteration diagnostics for the 2D case across 12 iterations. 
Prediction accuracy improves markedly (RMSE $1.113\!\to\!0.029$, $\!-97\%$; MAE $0.678\!\to\!0.017$) while empirical coverage rises from $37\%$ to $100\%$, yielding a small calibration error $|\widehat{\mathrm{PICP}}_{0.95}-0.95|\approx 0.05$ after iteration 4. 
Mean predictive uncertainty $\bar\sigma$ drops from $0.177$ to $0.052$ (roughly $70\%$), which translates into a $95\%$ mean interval width (MPIW) shrinking from $\sim0.694$ to $\sim0.204$. 
The safe set expands overall by $+27.7\%$, with small oscillations as the controller tracks targets near the safe frontier. 
The acquisition balance, measured by $\rho_k=\beta\sigma_k/|\mu_k|$, peaks early (iteration 3, $\rho_3\approx 0.32$) indicating variance-driven choices, and then remains below $0.05$ (iterations 7–12), signalling a shift to exploitation as posterior variance contracts.

Table~\ref{tab:2d-derived-metrics} summarizes the additional quantitative indicators extracted at each iteration of the safe exploration loop. 
The calibration error $|\widehat{\mathrm{PICP}}_{0.95}-0.95|$ measures the deviation between the empirical coverage of $95\%$ predictive intervals and their nominal level, thus assessing the reliability of the GP uncertainty estimates. 
A value close to zero indicates that predicted variances are well calibrated. 
The mean prediction‐interval width (MPIW) is computed as $3.92\,\bar\sigma$, representing the average width of the GP $95\%$ confidence bands; it captures the sharpness of the posterior. 
The pair $(\mu_k,\sigma_k)$ reports the predicted mean and standard deviation at the exploration target selected at iteration~$k$, while the exploration ratio
\[
\rho_k=\tfrac{\beta\,\sigma_k}{|\mu_k|}
\]
(with $\beta=2.5373$) quantifies the balance between exploration (variance‐driven) and exploitation (mean‐driven) behavior in the Upper Confidence Bound (UCB) acquisition function. 
Finally, $\|x_k^\star-x_{k-1}\|$ denotes the Euclidean distance between successive exploration targets, providing a measure of how aggressively the algorithm moves across the state space.

\medskip
\noindent
Over the twelve iterations, the safe set expands by approximately $+27.7\%$, while MPIW decreases by about $70\%$, indicating a simultaneous increase in confidence and precision of the learned model. 
The calibration error drops from $0.58$ to $0.05$, confirming that the predictive variance becomes well aligned with empirical errors. 
The exploration ratio $\rho_k$ peaks at $0.32$ in the early stage (iteration~3), when uncertainty still dominates, and stabilizes below $0.05$ afterwards, reflecting a progressive shift toward exploitation as the GP converges. 
Together, these metrics demonstrate that the framework not only improves accuracy and calibration but also expands the safe operational domain without constraint violations.

\begin{table}[t]
  \centering
  \caption{Derived per-iteration metrics for the 2D case. 
  Calibration error is $|\widehat{\mathrm{PICP}}_{0.95}-0.95|$, 
  MPIW is the mean $95\%$ interval width ($3.92\,\bar\sigma$), and 
  $\rho_k=\beta\,\sigma_k/|\mu_k|$ with $\beta=2.5373$.}
  \label{tab:2d-derived-metrics}
  \setlength{\tabcolsep}{2pt}
  \small
  \begin{tabular}{r r r r r r r}
    \toprule
    \textbf{Iter} &  \textbf{Cal. Err.} & \textbf{MPIW} & $\mu_k$ & $\sigma_k$ & $\rho_k$ & $\|x^\star_k-x_{k-1}\|$ \\
    \midrule
     1  &  0.58 & 0.694 & 0.2757 & 0.0039 & 0.036 & 2.002 \\
     2  &  0.58 & 0.694 & 0.2774 & 0.0033 & 0.030 & 1.695 \\
     3  &  0.58 & 0.694 & 0.2161 & 0.0274 & 0.322 & 2.872 \\
     4  &  0.05 & 0.459 & 0.2791 & 0.0025 & 0.023 & 2.355 \\
     5  &  0.05 & 0.459 & 0.2807 & 0.0020 & 0.018 & 2.639 \\
     6  &  0.05 & 0.459 & 0.2785 & 0.0026 & 0.024 & 2.137 \\
     7  &  0.05 & 0.328 & 0.2817 & 0.0015 & 0.014 & 2.228 \\
     8  &  0.05 & 0.328 & 0.2815 & 0.0017 & 0.015 & 1.957 \\
     9  &  0.05 & 0.328 & 0.2748 & 0.0043 & 0.040 & 2.976 \\
    10  &  0.05 & 0.204 & 0.2817 & 0.0014 & 0.013 & 1.721 \\
    11  &  0.05 & 0.204 & 0.2739 & 0.0046 & 0.043 & 2.320 \\
    12  &  0.05 & 0.204 & 0.2779 & 0.0029 & 0.027 & 2.421 \\
    \bottomrule
  \end{tabular}
\end{table}

\begin{figure}[t]
  \centering
  \includegraphics[width=\columnwidth]{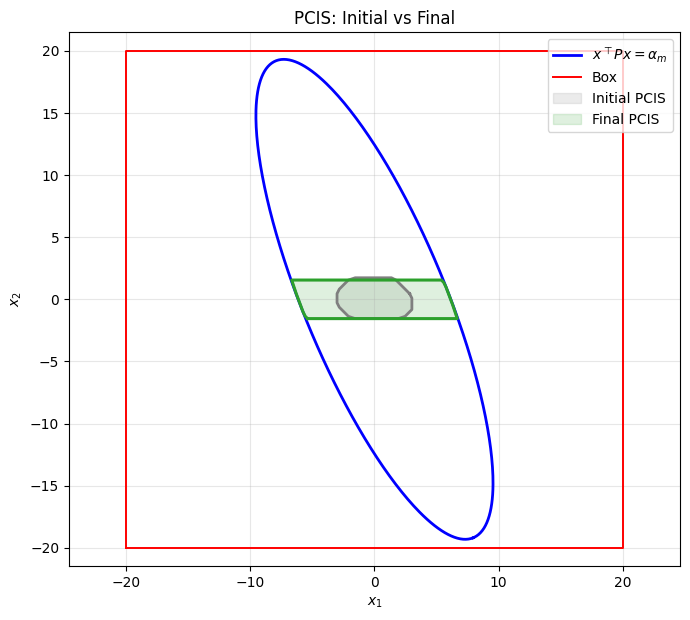}
  \caption{PCIS expansion in the 2D case. The blue ellipse shows the Lyapunov level set
  $x^\top P x=\alpha_m$; the red rectangle is the state box constraints. The shaded regions depict the
  initial (light gray) and final (green) certified PCIS.}
  \label{fig:pcis-initial-final}
\end{figure}

Figure~\ref{fig:pcis-initial-final} compares the initial and final certified probabilistic control–invariant sets (PCIS) in the 2D benchmark. As the GP posterior variance contracts (\(\bar\sigma: 0.177 \rightarrow 0.052\)), the certified safe set expands by about \(+27.7\%\) (Table~\ref{tab:iter-metrics-2d}). Throughout all 12 iterations, no state or input constraints were violated; every closed-loop trajectory remained inside both the state box and the certified PCIS.

These trends explain the observed PCIS enlargement: as uncertainty tightens, the Lyapunov-based safety margin reduces, certifying a larger admissible region without compromising the zero-violation record.

\subsection{Three Tank process}\label{sec:results:tank}
A realistic MIMO plant with water level constraints, valve saturations and cross-couplings stress tests certification tractability and real-time control.

For the interconnected three-tank process with MIMO actuation (three valves) we fixed level bands \([45\%,87\%]\) respectively LL low-low and HH high-high level as water level constraints in each tank. We operate around a stabilizable linearization and model unmatched nonlinearities with GP residuals. Sampling time \(T_s = 1\,\mathrm{s}\). We collect 200 samples initial data while the system was running around steady state for initial identification.
Exact GP vs.\ sparse FITC-SGPR with \(M\in\{10,20,30,40\}\) is tested to have a clear measurement of performance, and the same kernels as in Sec.~\ref{sec:kernel-comparison} are tested.
 We report RMSE, empirical coverage of nominal \(95\%\) predictive intervals, and training time. 

\subsection{GP modeling}\label{subsec:tank_gp}
From Tank 1, Tank 2, and Tank 3 residual channels results (Table~\ref{tab:wt-random-full}), it is possible to notice that Matérn kernels (\(\nu\in\{3/2,5/2\}\)) offer the best trade-off across channels, with coverage close to \(95\%\).

\begin{table}[ht]
\centering
\caption{Three–Tank kernel comparison. Best RMSE in \textbf{bold}.}
\label{tab:wt-random-full}
\setlength{\tabcolsep}{2pt}
\small
\begin{tabular}{lcccccc}
\toprule
& \multicolumn{2}{c}{\textbf{Tank 1}} & \multicolumn{2}{c}{\textbf{Tank 2}} & \multicolumn{2}{c}{\textbf{Tank 3}} \\
\cmidrule(lr){2-3}\cmidrule(lr){4-5}\cmidrule(lr){6-7}
Kernel & RMSE & Cov.\,(\%) & RMSE & Cov.\,(\%) & RMSE & Cov.\,(\%) \\
\midrule
RBF           & 0.52 & 99.6 & 0.77 & 54.3 & 0.43 & 97.1 \\
Matern32      & 0.41 & 99.9 & 0.40 & 99.9 & 0.25 & 100.0 \\
Matern52      & 0.44 & 99.5 & \bfseries 0.38 & 93.6 & \bfseries 0.21 & 100.0 \\
RBF+Matern32  & 0.44 & 98.4 & 0.56 & 86.9 & 0.45 & 95.6 \\
Periodic+RBF  & \bfseries 0.40 & 98.9 & 0.55 & 73.2 & 0.30 & 100.0 \\
Linear+RBF    & 0.65 & 95.1 & 1.04 & 86.9 & 0.91 & 64.4 \\
RBF–ARD       & 0.60 & 97.6 & 0.71 & 74.7 & 0.36 & 99.9 \\
Poly deg.\,2  & 2.95 & 94.8 & 0.86 & 81.5 & 0.62 & 88.1 \\
\bottomrule
\end{tabular}
\end{table}

Over-smooth RBF shows under-dispersion on Tank 2 (54.3\% coverage), which is problematic for safety certificates relying on GP uncertainty envelopes.

\subsection{Full GP versus Sparse GP}\label{subsec:tank_full_vs_sparse}
Exact GP  vs.\ FITC–SGPR ($M{\in}\{10,20,30,40\}$) aggregated over residual channels.
In aggregate, the exact GP achieves the lowest RMSE (0.8971) and perfect 100\% coverage, but with prohibitive latency (66.5\,s/query). Sparse FITC (SGPR) with \(M\in\{20,30,40\}\) delivers a near-nominal coverage of \(\approx 94.3\text{--}94.5\%\) and RMSE within \(\sim 10\%\) of the full GP (0.983–0.984), while enabling \(>10^3\times\) faster inference (e.g., \(66.5/0.026 \!\approx\! 2.6\times 10^3\)). Per-channel behavior is heterogeneous: for Tank 1, SGPR with \(M\ge 20\) \emph{improves} RMSE over the full GP (0.59 vs.\ 1.15) at equal coverage; for Tank 2, SGPR degrades both RMSE and coverage; for Tank 3, SGPR worsens RMSE (1.11 vs.\ 0.21) but maintains 100\% coverage for \(M\ge 20\). Under-dispersion at small \(M\) (e.g., \(M{=}10\), 82–93\% coverage) motivates a light variance calibration or a modest increase of the confidence scale \(\beta_t\) before PCIS. Overall, \(M\in[20,30]\) is a balanced operating point for online certification and control, preserving coverage close to the target, keeping RMSE acceptable, and reducing latency to the millisecond range. Given the mixed per-channel trends, a practical refinement is to allocate inducing points per residual channel and, if needed, apply heteroscedastic noise or channel-wise \(\beta_t\) to align empirical coverage with the nominal level used by PCIS.

\begin{table}[ht]
\centering
\caption{Three-tank full GP vs.\ sparse GP (aggregate over residuals). Coverage from nominal $95\%$ intervals.}
\label{tab:wt_sgpr_agg}
\setlength{\tabcolsep}{2pt}
\begin{tabular}{lccccr}
\toprule
Model & RMSE &  Cover.\ (\%) & Time (s) & $M$\\
\midrule
Full GP              & 0.8971  & 100.0 & 66.513 & 200\\
Sparse GP ($M{=}10$) & 1.2401  & 82.0  & 0.028  & 10\\
Sparse GP ($M{=}20$) & 0.9836  & 94.5  & 0.028  & 20\\
Sparse GP ($M{=}30$) & 0.9835  & 94.3  & 0.026  & 30\\
Sparse GP ($M{=}40$) & 0.9840  & 94.3  & 0.028  & 40\\
\bottomrule
\end{tabular}
\end{table}

Table~\ref{tab:wt_sgpr_perres} details Tank 1 – Tank 3. We therefore prioritise RMSE and coverage for model choice. SGPR with $M{=}20$–$30$ plus variance calibration (or slightly larger $\beta_t$) before PCIS.

\begin{table}[ht]
\centering
\caption{Three–Tank per-channel comparison (Tank 1–3).}
\label{tab:wt_sgpr_perres}
\setlength{\tabcolsep}{3pt}
\begin{tabular}{lcccr}
\toprule
\multicolumn{5}{c}{\textbf{Tank 1}}\\
\midrule
Model & RMSE & Cover.\ (\%) & Time (s) & $M$\\
\midrule
Full GP              & 1.1531 & 100.0 & 98.528 & 200\\
Sparse GP ($M{=}10$) & 1.0329 & 80.5  & 0.037  & 10\\
Sparse GP ($M{=}20$) & 0.5946 & 100.0 & 0.032  & 20\\
Sparse GP ($M{=}30$) & 0.5923 & 100.0 & 0.031  & 30\\
Sparse GP ($M{=}40$) & 0.5939 & 100.0 & 0.034  & 40\\
\midrule
\multicolumn{5}{c}{\textbf{Tank 2}}\\
\midrule
Model & RMSE & Cover.\ (\%) & Time (s) & $M$\\
\midrule
Full GP              & 1.3313 & 100.0 & 73.944 & 200\\
Sparse GP ($M{=}10$) & 1.4167 & 72.5  & 0.025  & 10\\
Sparse GP ($M{=}20$) & 1.2437 & 83.5  & 0.030  & 20\\
Sparse GP ($M{=}30$) & 1.2447 & 83.0  & 0.028  & 30\\
Sparse GP ($M{=}40$) & 1.2447 & 83.0  & 0.028  & 40\\
\midrule
\multicolumn{5}{c}{\textbf{Tank 3}}\\
\midrule
Model & RMSE & Cover.\ (\%) & Time (s) & $M$\\
\midrule
Full GP              & 0.2067 & 100.0 & 27.068 & 200\\
Sparse GP ($M{=}10$) & 1.2706 & 93.0  & 0.022  & 10\\
Sparse GP ($M{=}20$) & 1.1124 & 100.0 & 0.021  & 20\\
Sparse GP ($M{=}30$) & 1.1135 & 100.0 & 0.020  & 30\\
Sparse GP ($M{=}40$) & 1.1135 & 100.0 & 0.021  & 40\\
\bottomrule
\end{tabular}
\end{table}

\subsection{PCIS certification}\label{subsec:tank_pcis}
CLF decrease tightened by GP–UCB margin; risk level $\varepsilon{=}0.060$, confidence scale $\beta{=}2.797$ for the initial certificate.

\begin{figure}[ht]
  \centering
  \includegraphics[width=1.0\linewidth]{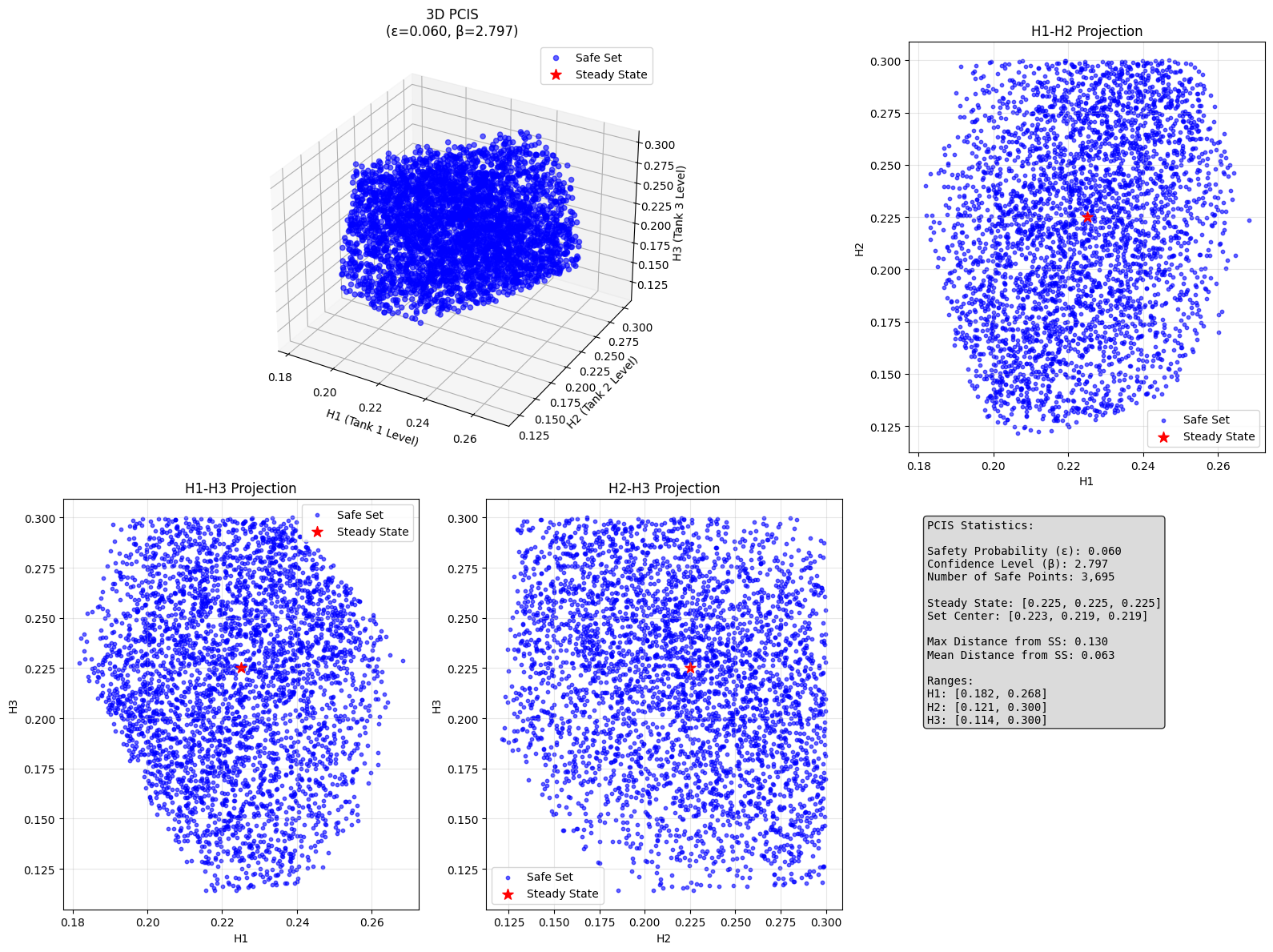}
  \caption{Three–Tank initial PCIS before safe exploration. Scatter projections and admissible axis ranges.}
  \label{fig:tank_PCIS}
\end{figure}

The initial certified ranges are
\[
H_1\in[0.182,\,0.268],\
H_2\in[0.121,\,0.300],\
H_3\in[0.114,\,0.300],
\]
with centre $\bar H\approx[0.223,\,0.219,\,0.219]$ and max distance $0.130$. The set is mildly anisotropic (wider along $H_2$–$H_3$ where GP uncertainty is lower). As data are added, variance contracts and the certified region expands.

\subsection{Unsafe vs. safe exploration}\label{subsec:tank_safe_vs_unsafe}
Without the safety layer, levels of water briefly leave the $[45\%,87\%]$ band: $H_3$ dips to $44.8\%$ at $t{\approx}131$s while the target is $45.2\%$ (Fig.~\ref{fig:tank_unsafe}) violating the LL constraints.

\begin{figure}[ht]
  \centering
  \includegraphics[width=1.0\linewidth]{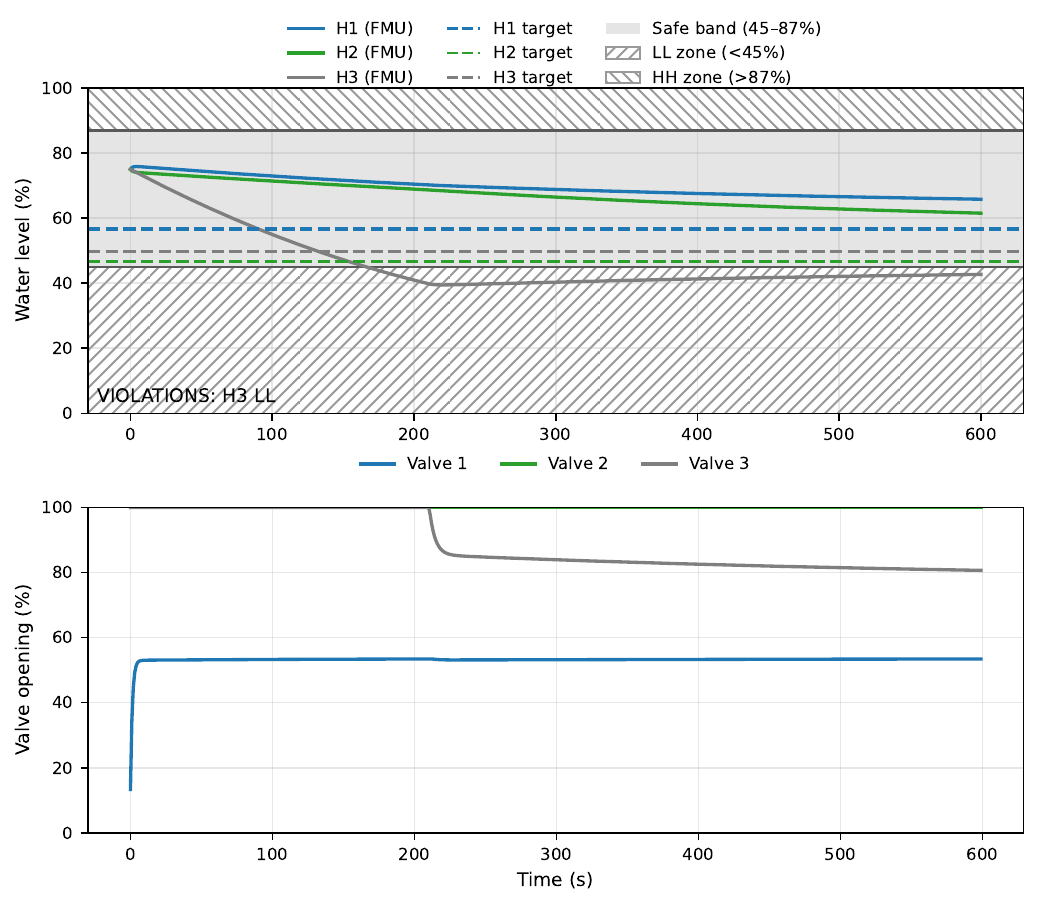}
  \caption{Three–Tank \emph{unsafe} exploration: top, level trajectories vs.\ band; middle, valve usage; bottom, event summary.}
  \label{fig:tank_unsafe}
\end{figure}

Across five iterations of our framework, \textbf{zero violations} are observed during an identification horizon $\approx\!2000$s (Fig.~\ref{fig:wt_progress_grid}).

\begin{figure*}[ht]
    \centering
    \includegraphics[width=\linewidth]{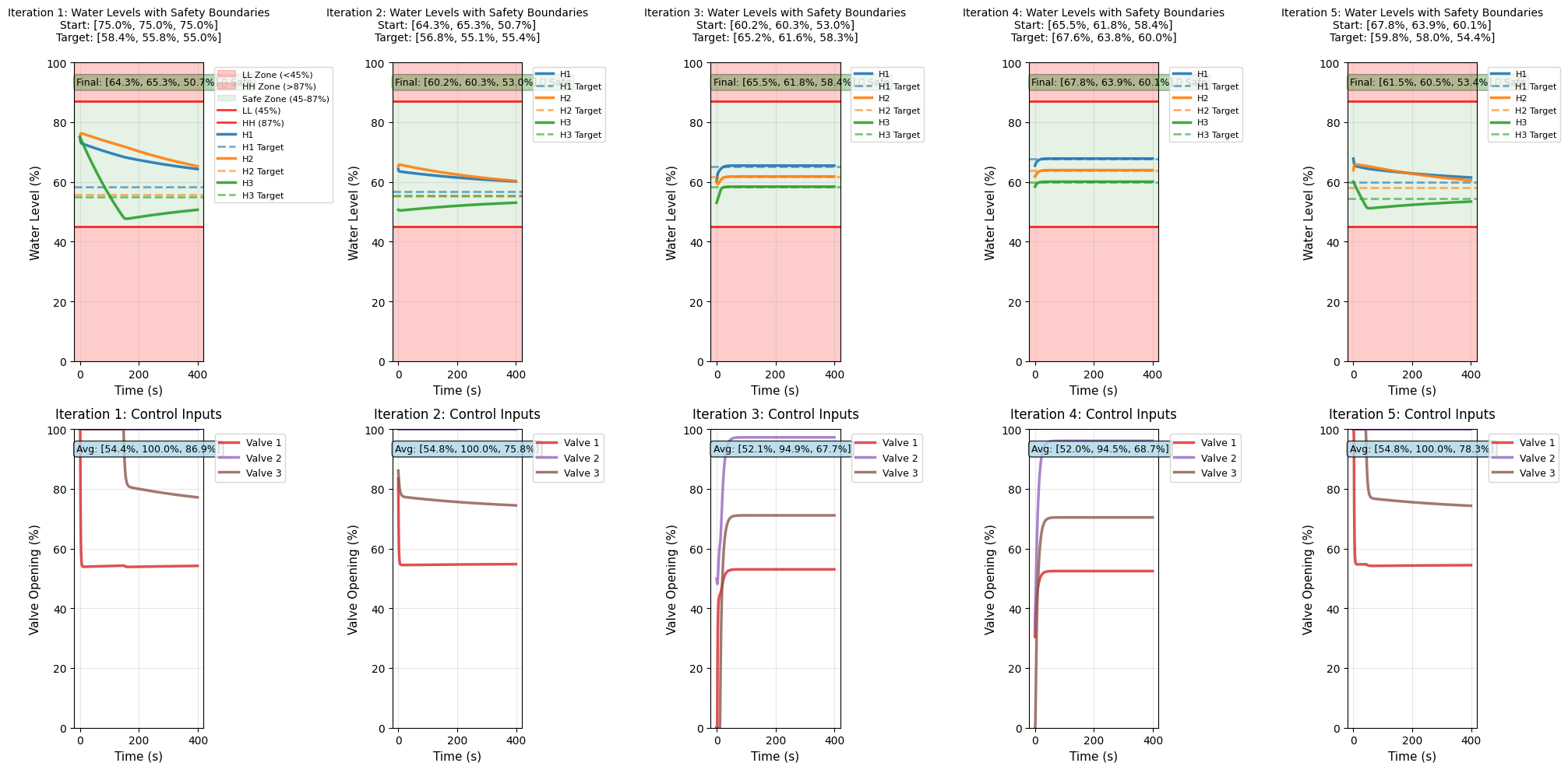}
    \caption{Three-Tank progressive safe exploration. Top: levels $H_{1:3}$ vs.\ bands; bottom: valve openings $\gamma_{1:3}$. Headers report start/target/final levels; bottom: average valve usage.}
    \label{fig:wt_progress_grid}
\end{figure*}

The main observation are the folowing:
\begin{enumerate}
    \item \textit{Safety.} All trajectories remain within $[45\%,87\%]$ throughout, without LL/HH violations.
    \item \textit{Convergence.} By iteration~3, steady plateaus match dashed references; iterations~4–5 retain small terminal errors while exploring nearby setpoints.
    \item \textit{Transients.} Limited overshoot and monotone settling ($\approx60$–$120$s for $H_1,H_2$). $H_3$ shows the largest initial correction in iter.~1, then smooth tracking from iter.~3 onward.
    \item \textit{Inputs.} Valve~2 saturates early to transfer mass to Tank~3, then relaxes as targets move closer and GP predictions improve; Valve~1 stabilises around mid-range; Valve~3 settles within $65$–$80\%$ depending on $H_3$ targets.
\end{enumerate}

\subsection{Discussion}
\label{sec:discussion}
The results demonstrate that combining GP residual modeling with PCIS enables informative exploration while maintaining strict adherence to safety constraints. On the low-dimensional polynomial system, the unsafe persistently exciting baseline achieves rapid learning but frequently violates state limits. In contrast, the GP--PCIS controller achieves comparable reduction in predictive uncertainty without any recorded violations, highlighting the efficacy of the safety certification in guiding exploration. On the three-tank process, control performance improves as the GP predictive variance decreases, allowing the certified safe set to expand and reducing the frequency of safety filter interventions. Median per-step QP solve times indicate that the proposed approach is feasible for real-time implementation.

\textbf{Limitations.} Despite these promising results, several limitations remain:  
(i) the PCIS certificate is conservative due to reliance on one-step worst-case bounds;  
(ii) exact GP training scales cubically with the dataset size, motivating the use of sparse or localized GP updates for high-dimensional or long-horizon tasks;  
(iii) the current guarantees are stepwise and do not explicitly account for cumulative risk over trajectories;  
(iv) safety critically depends on properly calibrated predictive uncertainty, as underestimated variance can yield overly optimistic safe sets.  

\textbf{Mitigation strategies.} These limitations can be addressed through several practical measures:  
conservative scheduling of the confidence parameter $\beta_t$, periodic empirical coverage checks to ensure accurate uncertainty quantification, variance calibration of the GP posterior, and simple tightening of reachability bounds to account for unmodeled effects. Together, these strategies improve robustness while preserving the benefits of safe exploration.

Overall, the results highlight that GP--PCIS provides a principled approach to balancing learning efficiency and safety, offering a flexible framework applicable to both low- and moderate-dimensional control tasks.

%
%
\section{Conclusion and Future Directions}
\label{sec:conclusion}

We have presented a practical framework for safe exploration of nonlinear systems, combining a nominal linear model with a GP residual embedded within a Lyapunov-based probabilistic control-invariant set, enforced through a compact control Lyapunov function quadratic program. Across both the polynomial benchmark and the three-tank process, the proposed approach enabled informative data collection without any observed safety violations, while progressively improving predictive accuracy at real-time computational cost.

\textbf{Future directions.} Several avenues exist to further improve the method. Conservatism could be reduced through multi-step chance constraints or reachability-based set tightenings. Scalability may be improved by employing incremental, sparse, or localized GP updates for high-dimensional or long-horizon tasks. Hardware validation, as well as integration with performance-oriented MPC, represent natural next steps toward real-world deployment.

\appendix

\section{Plant: Three--Tank Water--Level Process}\label{sec:plant}

\subsection{Decision and topology}\label{subsec:plant_topology}
We adopt a three-tank, cross–coupled benchmark. Tanks \(i\in\{1,2,3\}\) have constant cross-section areas \(A_i>0\) and liquid heights \(h_i(t)\) (states). Each tank drains to a reservoir through a controllable outlet valve with opening \(v_i\in[0,1]\) (inputs). Tanks are hydraulically coupled by lateral orifices between \((1,2)\) and \((2,3)\). A pump injects an exogenous inflow into tank~2, acting as a disturbance.

\paragraph*{State and input vectors.}
\[
x \;=\; \begin{bmatrix} h_1 & h_2 & h_3 \end{bmatrix}^\top,\qquad
u \;=\; \begin{bmatrix} v_1 & v_2 & v_3 \end{bmatrix}^\top.
\]

\subsection{Nonlinear mass–balance model}\label{subsec:nonlinear_model}
Let \(g\) be gravitational acceleration and define the following orifice flow laws (Torricelli):
\begin{align}
q_{i\downarrow}(h_i,v_i) &= v_i\, c_{i\downarrow}\,\sqrt{\max\{h_i,0\}}, \qquad c_{i\downarrow} := C_{d,i}\,a_{o,i}\sqrt{2g}, \label{eq:outflow}\\
q_{ij}(h_i,h_j) &= c_{ij}\,\sqrt{\max\{h_i-h_j,\,0\}}, \qquad c_{ij} := C_{d,ij}\,a_{ij}\sqrt{2g}, \label{eq:interflow}
\end{align}

with discharge coefficients \(C_{d,\cdot}\) and effective areas \(a_{o,i},a_{ij}\).
Let \(q_p(t)\) be the pump inflow to tank~2 (disturbance). The continuous–time dynamics are

\begin{align}
\dot h_1 &= \frac{1}{A_1}\big( q_{21}(h_2,h_1) - q_{12}(h_1,h_2)
            - q_{1\downarrow}(h_1,v_1) \big), \label{eq:h1}\\[2pt]
\dot h_2 &= \frac{1}{A_2}\big( q_p(t) + q_{12}(h_1,h_2) - q_{21}(h_2,h_1) \notag\\
        &\qquad\quad {}+ q_{32}(h_3,h_2) - q_{23}(h_2,h_3)
            - q_{2\downarrow}(h_2,v_2) \big), \label{eq:h2}\\[2pt]
\dot h_3 &= \frac{1}{A_3}\big( q_{23}(h_2,h_3) - q_{32}(h_3,h_2)
            - q_{3\downarrow}(h_3,v_3) \big). \label{eq:h3}
\end{align}

Equations \eqref{eq:outflow}–\eqref{eq:h3} define the truth plant. The nominal linear model used for control is a stabilisable linearisation of \eqref{eq:h1}–\eqref{eq:h3} about an operating point \((h^\star,v^\star)\); unmodelled effects (e.g., \(C_d\) drift, unmeasured leaks) are captured by the GP residual.

\subsection{Sampling and discretisation}\label{subsec:sampling}
The controller runs with zero–order hold and sampling time \(\Delta t={0.5}{s}\) (set to {1.0}{s} if needed for your hardware). Discrete dynamics for simulation are obtained via RK4 on \eqref{eq:h1}–\eqref{eq:h3}. For design, we use the discrete nominal pair \((A_d,B_d)\) from linearising and discretising the physics at \((h^\star,v^\star)\) with ZOH.

\subsection{Inputs/outputs and sensors}\label{subsec:io}
Inputs are valve openings \(v_i\in[0,1]\) (dimensionless).  
Measured outputs are water levels  heights \(y=h+\nu\), where \(\nu_k\sim\mathcal N(0,R)\) models sensor noise. Three sensors are available, one per each tank \(y=[h_1,h_2,h_3]^\top\).

\subsection{Operational constraints}\label{subsec:constraints}
Let safety bands for levels be \([h_i^{\min},h_i^{\max}]\) with \(0<h_i^{\min}<h_i^{\max}\).
\begin{align}
\mathcal X &:= \big\{ h\in\mathbb R^3 \ \big|\ h_i^{\min}\le h_i \le h_i^{\max},\ i=1,2,3 \big\}, \\
\mathcal U &:= \big\{ v\in\mathbb R^3 \ \big|\ 0 \le v_i \le 1,\ i=1,2,3 \big\}.
\end{align}
Chance constraints used in the controller read
\begin{equation}
\Pr\!\big(h_k \in \mathcal X\big) \ge 1-\delta,\qquad \forall k\ge 0,
\end{equation}
with \(\delta\in(0,1)\) allocated as described in Section~\ref{subsec:guarantee}. Use the same \([h_i^{\min},h_i^{\max}]\) in text, figures and captions.

\subsection{Disturbance and noise models}\label{subsec:disturbance}
The pump inflow is modelled as
\begin{equation}
\begin{aligned}
q_p(k) &= \bar q_p\,(1+\epsilon_k) + d_k, 
\quad &&\epsilon_k \sim \mathcal N(0,\sigma_\epsilon^2),\\
d_{k+1} &= \rho_d\, d_k + \omega_k, 
\quad &&\omega_k \sim \mathcal N(0,\sigma_d^2).
\end{aligned}
\label{eq:qp_disturbance}
\end{equation}

i.e., a biased, slowly varying AR(1) disturbance. Sensor noise \(\nu_k\sim\mathcal N(0,R)\) with \(R=\mathrm{diag}(\sigma_{h_1}^2,\sigma_{h_2}^2,\sigma_{h_3}^2)\). Any unmodeled hydraulic effects are absorbed by the GP residual in the controller.

\subsection{Initial conditions}\label{subsec:initial}
Unless stated otherwise, experiments start from
\begin{equation}
\begin{aligned}
h(0) &= \begin{bmatrix} h_1^0 & h_2^0 & h_3^0 \end{bmatrix}^\top,\\
h_i^0 &\in \big[\, h_i^{\min} + \alpha_i \big(h_i^{\max}-h_i^{\min}\big) \,\big],\quad
\alpha_i \in (0,1).
\end{aligned}
\label{eq:init_levels}
\end{equation}
We used \(h^0=\)[{0.22}{m}, {0.22}{m}, {0.22}{m}] in our experiments.

\begin{table}[t]
\centering
\caption{Three-tank plant parameters (symbols, meaning, units, values).}
\label{tab:plant_params}
\resizebox{\columnwidth}{!}{%
\begin{tabular}{llll}
\toprule
Symbol & Meaning & Units & Value \\
\midrule
$A_1,A_2,A_3$      & Tank cross--section areas     & m$^2$          & $0.015,\;0.015,\;0.015$ \\
$a_{o,1:3}$        & Outlet orifice areas          & m$^2$          & $5.0\times10^{-5}$ (each) \\
$a_{12},a_{23}$    & Inter--tank orifice areas     & m$^2$          & $3.0\times10^{-5},\;3.0\times10^{-5}$ \\
$C_{d,1:3}$        & Outlet discharge coeffs.      & --             & $0.62$ (each) \\
$C_{d,12},C_{d,23}$& Inter--tank discharge coeffs. & --             & $0.62,\;0.62$ \\
$c_{i\downarrow}$  & $C_{d,i}a_{o,i}\sqrt{2g}$     & m$^{5/2}$/s    & $1.373\times10^{-4}$ (each) \\
$c_{12},c_{23}$    & $C_{d,ij}a_{ij}\sqrt{2g}$     & m$^{5/2}$/s    & $8.24\times10^{-5},\;8.24\times10^{-5}$ \\
$h_i^{\min},h_i^{\max}$ & Safety bands             & m              & $0.12,\;0.30$ (each) \\
$r_{\max}$         & Max valve rate                & s$^{-1}$       & $1.0$ \\
$\Delta t$         & Sampling time                 & s              & $1$ \\
$\bar q_p$         & Mean pump inflow              & m$^3$/s        & $1.5\times10^{-5}$ \\
\bottomrule
\end{tabular}}
\end{table}

%
%
\bibliographystyle{plain}
\bibliography{bibliografia}

\end{document}

%% file: header.tex
\newcommand{\1}{\mathbbm{1}}

\makeatletter
\newsavebox\myboxA
\newsavebox\myboxB
\newlength\mylenA

\usepackage{tcolorbox}
\usepackage{pifont}
\definecolor{mydarkgreen}{RGB}{39,130,67}

\definecolor{mydarkred}{RGB}{192,47,25}

\newcommand*\overbar[2][0.75]{%
    \sbox{\myboxA}{$\m@th#2$}%
    \setbox\myboxB\null
    \ht\myboxB=\ht\myboxA%
    \dp\myboxB=\dp\myboxA%
    \wd\myboxB=#1\wd\myboxA
    \sbox\myboxB{$\m@th\overline{\copy\myboxB}$}
    \setlength\mylenA{\the\wd\myboxA}
    \addtolength\mylenA{-\the\wd\myboxB}%
    \ifdim\wd\myboxB<\wd\myboxA%
       \rlap{\hskip 0.5\mylenA\usebox\myboxB}{\usebox\myboxA}%
    \else
        \hskip -0.5\mylenA\rlap{\usebox\myboxA}{\hskip 0.5\mylenA\usebox\myboxB}%
    \fi}
\makeatother

\usepackage[colorinlistoftodos,bordercolor=orange,backgroundcolor=orange!20,linecolor=orange,textsize=scriptsize]{todonotes}

\newcommand{\rob}[1]{}

\usepackage{mdframed} 
\usepackage{thmtools}

\definecolor{shadecolor}{gray}{0.90}
\declaretheoremstyle[
headfont=\normalfont\bfseries,
notefont=\mdseries, notebraces={(}{)},
bodyfont=\normalfont,
postheadspace=0.5em,
spaceabove=1pt,
mdframed={
  skipabove=8pt,
  skipbelow=8pt,
  hidealllines=true,
  backgroundcolor={shadecolor},
  innerleftmargin=4pt,
  innerrightmargin=4pt}
]{shaded}

\declaretheorem[style=shaded,within=section]{definition}
\declaretheorem[style=shaded,sibling=definition]{theorem}
